\documentclass{article}
\usepackage{amsmath, amssymb, amsthm, tikz, subcaption, graphicx, authblk}
\usetikzlibrary{math,calc, positioning, arrows.meta,bending}
\usepackage[T1]{fontenc}
\usepackage[utf8]{inputenc}
\newtheorem{lemma}{Lemma}
\newtheorem{observation}{Observation}
\newtheorem{definition}{Definition}

\newcommand{\set}[1]{\left\{#1\right\}}

\sfcode`\.=1001
\sfcode`\?=1001
\sfcode`\!=1001
\sfcode`\:=1001

\newcommand{\ondemservice}{demand-responsive}

\newcommand{\relGap}{relative gap from optimum supply}
\newcommand{\RelGap}{Relative gap from optimum supply}
\newcommand{\R}{\mathbb R}
\newcommand{\Rp}{\mathbb R_{\ge 0}}

\newcommand*{\define}{:=}
\renewcommand*{\vec}[1]{\textbf{#1}}
\DeclareMathOperator*{\argmax}{argmax}

\usepackage{biblatex}
\usepackage{hyperref}
\hypersetup{
    colorlinks=true,
    linkcolor=blue,
    filecolor=magenta,      
    urlcolor=cyan,
    citecolor=red,
    pdfencoding=auto,
    }

\usepackage{cleveref}

\graphicspath{{figures/}}

\title{Staff Scheduling for Demand-Responsive Services}

\author[1]{Debsankha Manik}
\author[2]{Rico Raber}
\affil[1,2]{MOIA GmbH, Hamburg, Germany}

\begin{document}
\maketitle

\begin{abstract}
    Staff scheduling is a well-known problem in operations research and finds its
    application at hospitals, airports, supermarkets, and many others. Its goal is
    to assign shifts to staff members such that a certain objective function,
    e.g. revenue, is maximized.  Meanwhile, various constraints of the staff
    members and the organization need to be satisfied.  Typically in staff
    scheduling problems, there are hard constraints on the \emph{minimum} number
    of employees that should be available at specific points of time. Often
    \emph{multiple hard constraints} guaranteeing the availability of specific
    number of employees with different \emph{roles} need to be
    considered.  
    
    Staff scheduling for \ondemservice{} services, such as, e.g., ride-pooling and ride-hailing
    services, differs in a key way from
    this: There are often no hard constraints on the minimum number of employees
    needed at fixed points in time. Rather, the number of employees working at
    different points in time should vary according to the demand at those points
    in time. Having too few employees at a point in time results in lost
    revenue, while having too many employees at a point in time results in
    \emph{not having enough} employees at other points in time, since the total
    personnel-hours are limited.  The objective is to maximize
    the \emph{total reward} generated over a planning horizon, given a monotonic
    relationship between the number of shifts active at a point in time and the
    \emph{instantaneous reward} generated at that point in time. This key difference
    makes it difficult to use existing staff scheduling algorithms for planning
    shifts in \ondemservice{} services.
    
    In this article, we present a novel approach for modelling and solving staff
    scheduling problems for \ondemservice{} services that optimizes for the
    relevant reward function. 
\end{abstract}

\section{Introduction}
Working time of the employees is a finite and valuable resource in most
organizations. Therefore, planning shifts for the staff members in an optimal
way is very important in a wide variety of industries, including transportation,
hospitality, manufacturing and retail.  Especially in settings where the demand
for the goods or services provided varies with time,  it is very important for
the shift plan to ensure that the total deployed workforce at a point of time
matches the demand well \cite{thompson1998labor_2}: planning too few employees at
a period of high demand leads to lost productivity, planning too many employees
during low demand results in an opportunity cost. 

The staff scheduling process often needs to take various constraints into
account. Some of these constraints may be a consequence of limited availability
of a physical resource, such as workspace and equipments
\cite{zhang2005equipment,berman1997scheduling}. Some constraints may be legal in
nature, e.g. specific amount of breaks between successive shifts, mandatory days
off etc. Some constraints may be due to the need to maintain high job
satisfaction among the staff \cite{thopsn1999labor_3,gordon2004improving}: 
Certain staff members may want to work only during the morning or night, for
example. Finding any shift plan subject to these constraints me be a difficult
task, let alone one \emph{maximizing} some reward function. 

Certain application areas might require specific additional constraints and
different objectives.  Staff scheduling for transportation services
\cite{ernst1998train} often involves spatial constraints on the starting and
ending location of the shifts and conformity with a time table of operation.
Nurse scheduling in hospitals as well as emergency services need to accommodate
different categories of personnel (able to perform different functions), as well
as fair distribution of day and night shifts \cite{aickelin2000genetic,
mohammadian2019scheduling}. Many industries need to accommodate part time
employees \cite{vakharia1992efficient} working only certain days of the week.

Integer programming was introduced as a method of solving staff scheduling
problems as early as 1954 \cite{dantzig1954comment}, 
and various integer programming formulations have been proposed since then
\cite{keith1979operator,jacobs1996overlapping,bechtold1990implicit,beaumont1997scheduling,brusco2001starting}.
The number of decision variables in integer programming models may be very
large in realistic settings \cite{glover1986general}, leading to very high
runtimes. Therefore, various metaheuristics based methods have been studied for
solving staff scheduling problems, such as genetic programming 
\cite{aickelin2000genetic, tanomaru1995staff}, tabu search
\cite{glover1986general, cavique1999subgraph} and simulated annealing
\cite{brusco1993simulated, ernst1998train}.

We refer the reader to the review articles
\cite{ernst2004survey,alfares2004survey} for comprehensive surveys of the staff
scheduling problem and its applications.

\section{Demand Modelling for Staff Scheduling}
\label{sec:demand-modelling}
Staff scheduling in various settings start with a \emph{demand modelling} step
\cite{ernst2004survey,thompson1998labor_2}, where for each time step, a required 
number of active shifts are calculated. We will henceforth refer to this quantity as
the \emph{desired supply}. This is computed by first estimating future demand
of the services provided by the organization. Then domain knowledge on how much demand
can be served by how much supply is utilized to convert the estimated
demand into the desired supply \cite{dantzig1954comment}. In the later 
stages of the shift scheduling process \cite[~p.\,6]{ernst2004survey}, where
shifts are planned while maximizing a certain objective, deviations from the desired 
supply are penalized. Demand modelling, i.e.
converting estimated demand to desired supply can be performed using
various approaches \cite{thompson1998labor_2}, including the three described below.
\begin{description} 
    \item[Productivity standards] assuming the desired supply at each time
    step vary linearly with the demand at that time step.  
    \item[Service standards] equating the desired supply to the number that achieves a 
    constant fraction of the demand being served at each time step.
    \item[Economic standards] assigning a cost to each unit of workforce deployed at each time step, and 
    assigning a reward to each unit of demand served at each time step, and choosing
    the supply that maximizes the reward minus the cost 
    at that time step.
\end{description}
We note that both the service standards approach and the economic standards approach 
assumes that given a supply, i.e. the number of active shifts, and a demand, it is possible to determine the 
amount of demand that can be served or the amount of reward that can be generated. 
Crucially, the demand modelling step is usually performed without utilizing any 
knowledge of the available workforce \cite{thompson1998labor_2}.

\section{Our Contribution: Integrating Demand Modelling and Shift Optimization into a Single Step}
In this article, we present a novel approach for modelling and solving staff
scheduling problems for \ondemservice{} services. We eliminate demand modelling
as a separate step and incorporate it into a mixed-integer program producing 
optimal shift plans. Our
approach maximizes the total reward (e.g.
revenue) accumulated over the planning horizon, given that the relationship between the 
total number of active shifts
at a point of time and the instantaneous reward is known and is concave. By eschewing the demand modelling step, our approach is able to achieve 
a higher value of total reward compared to what is possible with approaches where demand
modelling is done in a separate step.

Our methods and results apply to a broad class of \ondemservice{} services
satisfying the following criteria. First, the number of employees available during
a \emph{planning horizon} (e.g. a week) is fixed and is known a priori, and so
are the numbers and lengths of shifts that must be assigned to each employee 
within the planning horizon. Second, there exists a metric whose sum total over
the planning horizon is to be maximized by the staff scheduling process (this
could be expected revenue, for example).  Third, the metric is concave with respect to
the number of active shifts at each point in time. However, this concave
function need not be the same for all points in time.

We present a way to model the staff scheduling problem as a mixed-integer convex optimization
problem. We also demonstrate
a concrete example of this general approach for planning driver shifts for an
\ondemservice{} mobility-as-a-service (MaaS) company. For 
sake of simplicity, we leave out the rostering step, where the planned shifts are 
assigned to individual employees, in this article. 

\section{Problem setting}
\label{sec:problem-setting}
For a natural number $n$, we denote by $[n]$ the set $\set{1,\dots,n}$.
We consider a planning horizon (e.g. 7 days) which is discretized into $T$ time steps.
We assume that at each time step, shifts of $k \ge 1$ \emph{different shift
types} can be started and each shift type $i\in [k]$ has a duration of $\delta_i \in
\mathbb{N}$ time steps. Shift types may indicate, for example, different employee groups.
As part of the shift planning process, we need to determine how many shifts of each type to start at which time step.
To this end, for every shift type $i\in [k]$ and time step $t\in [T]$, we introduce a decision variable 
\begin{equation}
x_{i,t}\define\text{ the number of shifts of type } i \text{ \textit{starting} at time } t.
\end{equation}

The number of \textit{active shifts} of type $i$ at time $t$ is then given by
\begin{equation}
y_{i,t} = \sum_{\tau=t-\delta_i+1}^{t} x_{i,\tau} \label{eq:active-shifts}
\end{equation}
and the total number of active shifts, or \emph{supply}, at time $t$ is given by
\begin{equation}
y_t = \sum_{i=1}^k y_{i,t}. \label{eq:total-active-shifts}
\end{equation}

Let the vector
\begin{equation}
\vec{x}\define (x_{i,t})_{i\in [k], t\in [T]} \label{eq:x-vector}
\end{equation}
describe the \textit{shift plan}. We assume that $\vec x$
is subject to linear constraints of the form $A\vec x \le \vec b$ for some matrix $A$ 
and vector $\vec b$ of corresponding dimensions. Some of these constraints are needed to specify that the
total working hours of the employees is a predetermined constant. Some rows of $A$ may indicate
certain operational constraints the staff schedule must fulfil (e.g., the total number
of simultaneous users of a physical resource must not be larger than the total number of
that resource available).

The objective is to find a feasible shift schedule that maximizes a reward
function over the planning horizon.  Specifically, for each time step $t\in [T]$,
we have a \textit{concave} function $f_{t} \colon \Rp 
\to \mathbb{R}$ mapping the total number of active shifts $y_t$ at time $t$ to the
\textit{instantaneous reward} $r_t$, 
\begin{align} 
    r_t &= f_t(y_t) \label{eq:def-reward-function}.
\end{align} 
The objective is to maximize the total reward over the
planning period $\sum_{t=1}^T r_t$.

Using \Cref{eq:active-shifts}, the resulting mixed-integer convex problem can then be formulated as 
\begin{align*} 
    \max \quad & \sum_{t=1}^T f_t\left(\sum_{i=1}^k y_{i,t}\right) \\ 
    \text{s.t.} \quad & y_{i,t} = \sum_{\tau=t-\delta_i+1}^{t} x_{i,\tau}  && \text{for all } i \in [k],\; t \in [T], \\ 
    & A\vec x \le \vec b \\ & x_{i,t} \in \mathbb{Z} && \text{for all } i \in [k],\; t \in [T]. 
\end{align*}

\section{Application to shift planning for on-demand mobility}\label{sec:application}
In on-demand mobility services such as MOIA, the number of vehicles deployed at
a given point in time should vary according to the demand from the customers to avail the 
service at that time. The total number of deployed vehicles naturally is
equal to the total number of driver shifts active at that time. The demand may
not be exactly known when the shifts need to be planned, but can be
estimated based on historical data. 

Usually, the objective function that should be maximized by the shift planning 
optimization process is the total number of served customers over the whole
planning period. Depending on the business model of the company, the objective
function may be slightly different, but it stands to reason that the objective
function is the sum over the time steps of a monotonically increasing function
of the number of active shifts at each time step: Because the more shifts are
active at a time step, the more customers can be served at that point of time.
In addition, we make the assumption that the function is concave. Intuitively, 
this means that the marginal benefit of adding one more shift at a time step
decreases with the number of active shifts at that time step. This is a reasonable
assumption, since the more shifts are active at a time step, the more likely it is
that the demand in the vicinity of a vehicle is already satisfied by another
vehicle, and thus, the contribution of that vehicle to the total number of served
customers is smaller.

Finally, we assume that the total number of drivers
employed by the company is a known and fixed number. 


\subsection{The Variables and Constraints}
Following the general problem setting in Section~\ref{sec:problem-setting}, we
assume we have $k\ge 1$ different shift types, each with a fixed duration of $\delta_i$
time steps, $i\in [k]$.  A shift type may specify various properties of the employees, e.g.
the contract type of the driver (full-time, part-time, contracted etc.), the
days of the week as well as the time of the day (day/night) the
driver is available.

Then the monotonically increasing concave function $f_{t}$ maps $y_{t}$, the total number of shifts active at
the time step $t$, to the instantaneous reward $r_t$ (e.g. number of trips served
at that time step). We assume that the function $f_t$ is known (e.g. by estimation
from historical data).

As described above, the decision variables are $x_{i,t}$, the number of shifts of type $i$
started at time $t$. Various legal and operational constraints may need to be considered
when optimizing the shifts, for example:

\begin{itemize}
    \item The total number of active shifts at a time may be bounded above by
    the total number of vehicles available.  
    \item The total number of shifts of a certain type within a day must be
    bounded above by the total number of drivers with that shift type employed by the provider.
\end{itemize}

\subsection{The shift types}
Let the planning horizon be a week, and suppose that shifts may be started at the 
beginning of every hour, i.e. the number of time steps is $T=7\times 24=168$.
For simplicity, we assume that there are only one kind of employees, 
each working in $s$ shifts of length $\delta$ hours each (typical examples are $s=5$ 
and $\delta=8$). Consequently, we have 
only one shift type of duration $\delta$ hours. Therefore $y_{t}=y_{1,t}$ and $x_t = x_{1,t}$. 
Let the number of employees be $N$. Then the total number of shifts within the planning 
horizon of one week is $sN$.

\subsection{The reward function}
It is reasonable to assume that the on-demand mobility provider is interested in 
maximizing the total number of served rides within the planning period. Then the 
monotonically increasing concave reward function $f_t$ should map the number of active 
shifts at time $t$ to the number of served rides at that time. 
The number of served rides, given the number of active shifts $y_t$, should depend on  
the number of \emph{demanded rides}: that is, $f_t$ should be parametrized by the
demanded rides $d_t$ at time $t$: $f_t=f_{t, d_t}$. Additionally, it stands to reason 
that, for a fixed number of active shifts, the number of served rides will
not decrease with the number of demanded rides, i.e.,
\begin{equation}
    \frac{\partial f_{t, d_t}(y_t)}{\partial d_t} \ge 0 \quad \text{for all } y_t\ge 0.
\end{equation}
Also, since the number of served rides cannot exceed the  number of demanded rides,
\begin{equation}
    \lim_{y_t\to\infty} f_{t, d_t} (y_t)\le d_t.
\end{equation}

The method of this article does not depend on the specific choice of $f_t$, so long as 
it is monotonically increasing and concave. For the sake of concreteness, we will
presently make the following choice:
\begin{equation}
    f_{t, d_t}(y_t) = d_t\left(1-e^{-ay_t/d_t}\right), \label{eq:chosen-reward-function}
\end{equation}
for some $a>0$; see \Cref{fig:reward-function} for a visual representation.

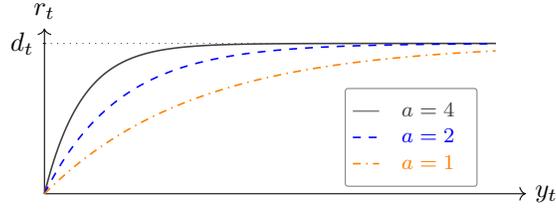
\begin{figure}[t]
\centering
\begin{tikzpicture}[scale=2]
    \draw[->] (0,0) -- (3.2,0) node[right] {$y_t$};
    \draw[->] (0,0) -- (0,1.1) node[above] {$r_t$};
    \draw (0.01,1) -- (-0.01,1) node[left] {$d_t$};
    \draw[dotted, darkgray] (0,1) -- (3,1);
    \draw[darkgray, semithick, domain=0:3,smooth, samples=100] plot (\x,{(1 - exp(-4*\x))});
    \draw[semithick, blue, dashed, domain=0:3,smooth, samples=100] plot (\x,{(1 - exp(-2*\x))});
    \draw[semithick, orange, dash dot, domain=0:3,smooth, samples=100] plot (\x,{(1 - exp(-1*\x))});
    \begin{scope}[scale=0.5, shift={(-1,-0.1)}]
    \draw[gray,thin, rounded corners=1] (5,0.2) rectangle (6.75,1.5);
    \draw[semithick, gray] (5.1, 1.2) -- (5.45, 1.2);
    \node[darkgray] at (6.1, 1.22) {\footnotesize $a=4$};
    \draw[semithick, dashed, blue] (5.1, 0.85) -- (5.45, 0.85);
    \node[blue] at (6.1, 0.87) {\footnotesize $a=2$};
    \draw[semithick, dash dot, orange] (5.1, 0.5) -- (5.45, 0.5);
    \node[orange] at (6.1, 0.52) {\footnotesize $a=1$};
    \end{scope}
    
\end{tikzpicture}
\caption{The reward $r_t$ as a function of supply $y_t$ as defined in \eqref{eq:chosen-reward-function} for 
$d_t = 1$ and different values of the parameter $a$.}
\label{fig:reward-function}
\end{figure}

\subsubsection{Demanded rides}
The only remaining unknown in the reward function is the demanded rides for each
time $t$. Since our approach does not assume any property of the demanded rides,
we will assume that $d_t$ is a periodic function of time $t$ with 
daily and weekly seasonality. More precisely, we will assume that $d_t$ is given by
the function
\begin{equation}
    d_t = \frac{d_{\max}}{2}\left[1-\cos\left(\frac{\pi t}{12}\right)\right]
    \sin\left(\frac{\pi t}{T}\right),
\end{equation}
see \Cref{fig:demanded-rides} for a visual representation.

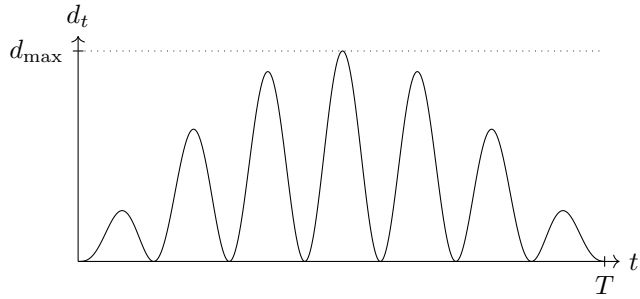
\begin{figure}[b]
\centering
\begin{tikzpicture}
    \draw[->] (0,0) -- (7.2,0) node[right] {$t$};
    \draw[->] (0,0) -- (0,3) node[above] {$d_t$};
    \draw (7,0.07) -- (7,-0.07) node[below] {$T$};
    \draw (0.07,2.8) -- (-0.07,2.8) node[left] {$d_{\max}$};
    \draw[dotted, darkgray] (0,2.8) -- (7,2.8);
    \draw[xscale=1/57,domain=0:57*7,smooth, samples=1000] plot (\x,{1.4*(1-cos(2*pi*\x))* sin(pi*\x/7)});
\end{tikzpicture}
\caption{The demanded rides $d_t$ as a function of time $t$.}
\label{fig:demanded-rides}
\end{figure}

\subsection{The constraints}
\subsubsection{Shift plan must match the number of employed drivers}
Recall that we have $N$ employees, each working $s$ shifts of $\delta$ hours each. This imposes 
the constraint that the total number of shifts within the planning horizon must be $sN$,
\begin{equation}
    \sum_{t=1}^T x_t = sN. \label{eq:constraints-total-shifts}
\end{equation}
However, this constraint alone is too weak: it does not take into account the fact that a 
single employee cannot work more than one shift at a time. In fact, most workplaces
have a policy that each employee needs to have a minimum break of, say, $\beta$ time steps 
between two consecutive shifts. We can accommodate these constraints by extending each shift by 
$\beta$ time steps and requiring that at any point in time, there
are at most $N$ \textit{extended shifts} active. More formally, we introduce auxiliary
variables $z_t$ for  $t\in [T]$, given by
\begin{equation}
    z_t \define \sum_{\tau=t-\delta-\beta +1}^{t} x_{\tau}, \label{eq:z-definition}
\end{equation}
and require them to satisfy the following constraints
\begin{equation}
    z_t \le N \quad \text{for all } t\in [T]. \label{eq:constraints-max-extended-shifts}
\end{equation}
We note that for all $t\in [T]$, $x_t \leq z_t$, and therefore, 
\eqref{eq:constraints-max-extended-shifts} implies that 
\begin{equation}
    x_t \le N \quad \text{for all } t\in [T]. \label{eq:constraints-max-shifts}
\end{equation}
We describe each extended shift $s$ by an interval $s = [t_1, t_2)$, 
where $t_1$ is the starting time of the shift and $t_2$ is $\beta$ time steps 
after the end of the shift, i.e. $t_2 = t_1 + \delta + \beta$. We say that 
two extended shifts $s_1$ and $s_2$ \textit{overlap} if $s_1\cap s_2 \neq \emptyset$.
\begin{observation}\label{lemma:graph-bipartite}
    Let $D_1$ and $D_2$ be two drivers and $S_1$ and $S_2$ be the set of extended shifts assigned
    to $D_1$ and $D_2$, respectively. Let $G$ be the graph with node set
    $S_1 \cup S_2$, where two extended shifts are connected by an edge iff they overlap.
    Then $G$ is bipartite with node sets $S_1$ and $S_2$ if and only if there is no overlap between any two 
    extended shifts assigned to the same driver.
\end{observation}

We can now prove that the constraints \eqref{eq:constraints-total-shifts} and
\eqref{eq:constraints-max-extended-shifts} are sufficient to ensure that the
shift plan can be assigned to the drivers in a feasible way. 

    

\begin{lemma}\label{lemma:shift-assignment}
If (\ref{eq:constraints-total-shifts}) and (\ref{eq:constraints-max-extended-shifts}) hold, then
there exists an assignment of shifts to drivers such that each driver works exactly 
$s$ shifts of $\delta$ hours each, and each driver has a break of at least $\beta$ time steps between
any two consecutive shifts.
\end{lemma}

\begin{proof}
We split this proof into two parts. First, we show that a shift plan
satisfying (\ref{eq:constraints-total-shifts}) and
(\ref{eq:constraints-max-extended-shifts}) can be assigned to $N$ drivers such that
there is a break of at least $\beta$ time steps between any two shifts of a 
driver. In the second part, we show that the same shift plan can be assigned to
$N$ drivers such that, additionally, each driver works exactly $s$ shifts of $\delta$ hours each.

Let $\vec{x}$ as defined in \eqref{eq:x-vector} be a shift plan that fulfills (\ref{eq:constraints-total-shifts}) and
(\ref{eq:constraints-max-extended-shifts}).

\paragraph*{Part 1: Proof that the shift plan can be assigned to the drivers
with adequate breaks between shifts\\} We note that, from \eqref{eq:z-definition},
\begin{equation}\label{eq:split_zt}
    z_t = \sum_{\tau=t-\delta-\beta +1}^{t} x_\tau
        =\sum_{\tau=t-\delta-\beta +1}^{t-1} x_\tau + x_t.
\end{equation}
The first term on the right-hand side of \eqref{eq:split_zt} is the number of drivers who may not be assigned any new
shift at time $t$, because they either have an ongoing shift or their last
shift ended less than $\beta$ time steps ago. The constraint $z_t\leq N$ therefore guarantees that 
at time $t$, there are at least $x_t$ drivers who may be assigned a new shift.

\paragraph*{Part 2: Proof that the shift plan can be assigned to the drivers
with each driver working exactly $s$ shifts\\}
    
Now assume we have an assignment of shifts to drivers such that each driver has
a break of at least $\beta$ time steps between two consecutive shifts, 
and assume that not every driver works exactly $s$ shifts.
Then, by \eqref{eq:constraints-total-shifts}, there must be a driver $D_1$ who works more than $s$ shifts and 
a driver $D_2$ who works less than $s$ shifts.
Let $S_1$ be the extended shifts of driver $D_1$ and $S_2$ be the extended shifts of driver $D_2$.
Then we have $|S_1| \ge |S_2| + 2$.

Let $G$ be the graph with node set $S_1 \cup S_2$, where two nodes are connected
by an edge iff they overlap. By Observation~\ref{lemma:graph-bipartite}, $G$ is bipartite
with node sets $S_1$ and $S_2$.  Since all extended shifts have the same length,
any extended shift can overlap with at most two other extended shifts.
Therefore, every node in $G$ has degree at most 2. As a consequence, every
connected component of $G$ is a path or a cycle. In particular, every connected
component contains a path with all nodes of that connected component. 

Since $G$
is bipartite, for any such path $P$, we have $|S_1 \cap P| - |S_2 \cap P| \in
\set{-1,0,1}$. Since $|S_1| \ge |S_2| + 2$, there must be a path $P$ such that
$|S_1\cap P|  - |S_2\cap P| = 1$.  Define $P_1 \define S_1 \cap P$ and $P_2
\define S_2 \cap P$. Since $P$ is a path containing nodes only from a single
connected component of $G$, there are no edges between $P$ and $(S_1 \cup S_2)
\setminus P$. In particular, no edge exists between 
\begin{itemize}
    \item $S_1 \setminus P$ and $P_2$, as well as
    \item $S_2 \setminus P$ and $P_1$.
\end{itemize}
As a consequence, we can reassign the extended shifts in $P_1$ to
$S_2$ and the extended shifts in $P_2$ to $S_1$, thereby creating the new assignments
$S_1'$ and $S_2'$ of shifts among $D_1$ and $D_2$, with
\begin{align*}
    S_1' &= (S_1 \setminus P_1) \cup P_2, \\
    S_2' &= (S_2 \setminus P_2) \cup P_1,
\end{align*}
satisfying the property that the resulting graph remains
bipartite. By Observation~\ref{lemma:graph-bipartite}, each
driver still has a break of at least $\beta$ time steps between any two shifts.
In this way, we have reduced the difference in number of shifts between $D_1$
and $D_2$ by 2.  We can repeat this process until every driver works exactly $s$
shifts.
\end{proof}

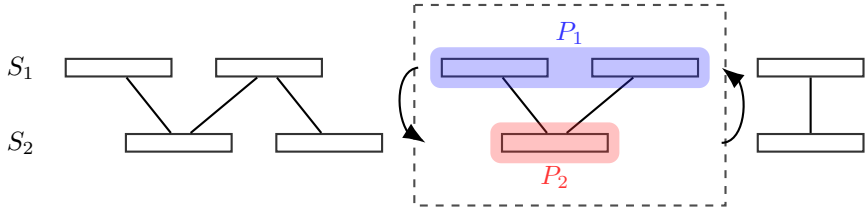
\begin{figure}
\tikzset{
    shift node/.style={
                      shape=rectangle, 
                      thick,
                      draw=black!80,
                      minimum width = 40pt,
                      anchor=center
                    },
}
\begin{tikzpicture}[thick]
\node at (-0.5,2.5) {$S_1$};

\foreach \x [count=\i] in {0,2, 5, 7 , 9.2}{
    \node[shift node] (a\i) at (1.0*\x+0.8,2.5) {};
}

\node at (-0.5,1.5) {$S_2$};

\foreach \x [count=\i] in {1,3,6,9.4}{
    \node[shift node] (b\i) at (1.0*\x+0.6,1.5) {};
}

\draw (a1) -- (b1);
\draw (a2) -- (b1);
\draw (a2) -- (b2);

\draw (a3) -- (b3);
\draw (a4) -- (b3);

\draw (a5) -- (b4);

\coordinate[xshift=-4, yshift=4] (a) at (a3.north west);
\coordinate[xshift=4, yshift=-4] (b) at (a4.south east);
\fill[fill=blue!80!white, fill opacity=0.3, rounded corners] (a) rectangle (b);
\node[yshift=13,text=blue!80!white] (P1) at ($(a3)!0.5!(a4)$) {$P_1$};

\coordinate[xshift=-4, yshift=4] (a) at (b3.north west);
\coordinate[xshift=4, yshift=-4] (b) at (b3.south east);
\fill[fill=red!80!white, fill opacity=0.3, rounded corners] (a) rectangle (b);
\node[yshift=-13,text=red!80!white] (P1) at (b3) {$P_2$};

\coordinate[xshift=-10, yshift=20] (a) at (a3.north west);
\coordinate[xshift=10, yshift=-20] (b) at (a4.north east |- b3.south east);
\draw[dashed, color=black!70!white] (a) rectangle (b);

\draw[-{Latex[length=3mm]}] [bend left=20] ([xshift=-3mm]a3.west) to [bend left=270]([xshift=-2mm]a3.west |- b3.west);
\draw[{Latex[length=3mm]}-, bend left=30] ([xshift=3mm]a4.east) to [bend left=90] ([xshift=3mm]a4.east |- b3.east);

\end{tikzpicture}
\caption{Proof of Part 2 of \Cref{lemma:shift-assignment}. The dashed rectangle 
shows a connected component of the bipartite graph $G$. $P_1$ can $P_2$ can be swapped,
keeping the graph bipartite, because there is no edge between the connected component 
and the rest of the graph.}
\label{fig:second-case}
\end{figure}

We note that the proof of \Cref{lemma:shift-assignment} is constructive, and
therefore, we can use it to assign the shifts to the drivers in a feasible way.

\subsubsection{Available vehicles constraint}
If the number of vehicles available to the mobility provider is fixed, say, $c$ then we need 
another constraint: The total number of active shifts at any time must be less than $c$,

\begin{equation}
    y_t \le c, \quad \text{for all } t\in [T]. \label{eq:constraints-vehicles}
\end{equation}

\subsection{The Full Program}
\label{sec:full-program}
The problem to maximize the total reward over the planning horizon
subject to shift and vehicle constraints can then be formulated as the following mixed-integer
convex program.
\begin{subequations}
\label{full-mip-ours}
\label{mip}
\begin{align}
     \max \quad & \sum_{t=1}^{T} f_t\left(y_t\right)\\
    \text{s.t.} \quad & y_{t} = \sum_{\tau=t-\delta + 1}^{t} x_{\tau}, && \text{for all } t\in [T], \label{program:first-constraint}\\
    &y_t \le c, &&\text{for all } t\in [T],  \\
    &\sum_{t=1}^T x_t = sN, \\
    &z_t = \sum_{\tau=t-\delta-\beta +1}^{t} x_t, &&\text{for all } t\in [T], \\
    &z_t \le N, &&\text{for all } t\in [T],  \\
    &x_t \ge 0 &&\text{for all } t\in [T], \\
    &x_t \in \mathbb{Z} &&\text{for all } t\in [T], \label{program:last-constraint}
\end{align}
\end{subequations}
where 
\begin{equation}
f_{t}(y_t) = d_t\left(1-e^{-ay_t/d_t}\right) \label{eq:concr-prodmodel} 
\end{equation}
and
\begin{equation}
d_t = d_{\max}\left[1+\sin\left(\frac{\pi t}{12}\right)\right].
\end{equation}

For a better overview, we summarize the parameters of the mathematical program 
in the following table.
\begin{center}
\begin{tabular}[!ht]{ll}
    Parameter & Meaning \\
    \hline
    $T$ & number of time steps in the planning horizon \\
    $N$ & number of employees \\
    $s$ & number of shifts each employee works \\
    $\delta$ & length of each shift  \\
    $\beta$ & minimum break between two consecutive shifts \\
    $d_{\max}$ & maximum amplitude of the sinusoidal demand pattern \\
    $a$ & steepness of the concave function $f_t$ \\
\end{tabular}
\end{center}

\subsection{Benchmarking against shift agnostic optimum}
\label{sec:shift-agnostic-optimum}
Maximizing the total reward is the objective of the shift planning process we have 
described. If more constraints are imposed, the total optimum reward may decrease, yet
those constraints may be necessary for operations, legal or strategic reasons.
For an \ondemservice{} service provider, it is important to know how much is the cost of such 
constraints. To that end, we will now present a novel metric for evaluating the quality 
of a shift plan.

\subsubsection*{The shift agnostic optimum}
Henceforth, we will assume that the reward functions $f_t$ are continuous for all $t\in [T]$.
We define the \emph{shift agnostic optimum} as the maximum reward that can be achieved,
ignoring all the constraints and even ignoring the fact that the shifts are of 
fixed lengths, but keeping the \emph{total personnel working time} fixed.
Specifically, the shift agnostic optimum is defined as the optimal value of the following convex optimization problem:
\begin{equation}
\begin{split}
    \max \quad & \sum_{t=1}^{T} f_t\left(y_t\right)\\
    \text{s.t.}\quad &\sum_{t=1}^T y_t = sN\delta, \\
    &y_t \ge 0 \quad \text{for all } t\in [T].
\end{split}
\label{problem:shift-agnostic-optimum}
\end{equation}
Note that since the objective function of \eqref{problem:shift-agnostic-optimum} is 
continuous and the feasible set is compact, the problem indeed has
an optimal solution and the shift agnostic optimum is well-defined.

In the following lemma, we provide optimality conditions for problem \eqref{problem:shift-agnostic-optimum} that will 
help us to derive explicit formulas for the shift agnostic optimum 
for the concrete example that we will study in the remainder of this paper.

\begin{lemma}
\label{lemma:exact-shift-agniostic-optimum}
Let the reward functions $f_t$ be continuously differentiable and concave for all $t\in [T]$.
Then a vector $(y^*_1, y^*_2,\cdots,y^*_T)$ is an optimal solution to the optimization 
problem~\eqref{problem:shift-agnostic-optimum} if and only if there exist constants 
$\lambda\in \R$ and $\mu_1,\dots, \mu_T\in \Rp$ such that the following conditions hold:
\begin{equation}
    \begin{aligned}
        \frac{\partial f_t(y^*_t)}{\partial y_t} + \mu_t&= \lambda, && \text{for all } t\in [T]\\
        \sum_{t=1}^T y^*_t &= sN\delta, \\
        y_t^* &\ge 0, && \text{for all } t\in [T],\\
        y_t^* \mu_t &= 0, && \text{for all } t\in [T].
    \end{aligned}
    \label{eq:shift-agnostic-optimum-exact}
\end{equation}
If in addition $f_t$ is strictly concave for all $t\in [T]$, then the optimal solution is unique.
\end{lemma}
\begin{proof}
The KKT conditions \cite{boyd2004convex} guarantee that any solution to the 
maximization problem \eqref{problem:shift-agnostic-optimum} must satisfy 
\eqref{eq:shift-agnostic-optimum-exact}.

Since $f_t$ is concave for all $t\in [T]$, and all constraints are linear, the problem
is convex, and thus, the above conditions also become 
sufficient. 

Finally, if $f_t$ is strictly concave in $y_t$ for all $t$, then the
objective function of problem \eqref{problem:shift-agnostic-optimum} is strictly concave, and thus, the optimal solution is unique.
\end{proof}

For the particular choice of $f_t$ given by
\eqref{eq:chosen-reward-function}, 
\begin{equation*}
    f_{t, d_t}(y_t) = d_t\left(1-e^{-ay_t/d_t}\right),
\end{equation*}
the shift agnostic optimum is given by
\begin{equation}
    \label{eq:shift-agnostic-optimum}
    y^*_t = \frac{sN\delta}{\sum_{\tau=1}^T d_t} d_t,
\end{equation}
which can be verified via equation system \eqref{eq:shift-agnostic-optimum-exact}.

In the subsequent sections, we will use the shift agnostic optimum as a benchmark to 
compare our approach against. To that end, we define the \emph{\relGap} as follows.

\begin{definition}[\RelGap{}]
    Let $\vec{x}$ be a shift plan with with total personnel time $sN\delta$, and let 
    $y_t$ be the corresponding supply at time $t$. Let $r^*$ be the corresponding shift agnostic 
    optimum, i.e., the optimum value of problem \eqref{problem:shift-agnostic-optimum}.
    Then the \emph{\relGap{}} $\Delta(\vec{x})$ is defined as the relative difference between
    the total reward of the shift plan $\vec{x}$ and the shift agnostic optimum,
    \begin{align}
            \Delta(\vec{x}) &\define \frac{ r^*  - \sum_{t=1}^T f_t(y_t)}{r^* }.
            \label{def:relgap}
    \end{align}
\end{definition}

If $r^* > 0$ and the reward functions $f_t$ are non-negative for all $t\in [T]$ (as it is the case for 
our particular choice of $f_t$ \eqref{eq:chosen-reward-function}), then 
$\Delta(\vec{x}) \in [0,1]$.

\subsection{Benchmarking against approaches using a separate demand modelling step}
\label{sec:separate-demand-modelling}
Here we outline a method to compare the quality of the shift plans generated by the process
described above with shift plans generated using a traditional demand modelling
step \cite{thompson1998labor_2}. As described in \Cref{sec:demand-modelling}, 
these approaches first convert the demand to required working staff
at each point of time. At the second step, the shift plans are generated by solving a 
mixed-integer convex problem that minimizes the sum of quadratic deviations from the 
required working staff at each point of time.

We will consider two different approaches for demand modelling described in Section 
\ref{sec:demand-modelling}: (a) 
service standards, and (b) economic standards, and compare the resulting shift plans 
from each approach with the approach we introduced. 

\subsubsection*{Service standards} 
In this approach, the number of required working staff at each point is time is
equated to the number that achieves a constant fraction of the demand being
served. Assuming that the reward function \eqref{eq:chosen-reward-function}
denotes the amount of served demand, maintaining a constant fraction $c\le 1$ of
the demand being served requires $y^{\text{s}}_t$ working staff at time $t$,
where 
\begin{equation}\label{eq:service-standards-implicit}
    f_t(y^{\text{s}}_t) = c d_t.
\end{equation}
Solving \eqref{eq:service-standards-implicit} for $y^{\text{s}}_t$ yields
\begin{equation} 
    y^{\text{s}}_t = a d_t \log{\frac{1}{1-c}}. \label{eq:service-standards}
\end{equation}


\subsubsection*{Economic standards} 
Maintaining a constant service standard may not be economically the best choice: A low service standard at low
demand periods may lead to less lost revenue than a similarly low service
standard at a high demand period. The economic standard approach attempts to tackle
this problem by assigning a cost per \emph{deployed staff}, as well as a reward (e.g. revenue)
per \emph{served demand} at each point in time. Then the number of required working 
staff at a point in time is obtained by the number of staff that maximizes the 
reward minus the cost, i.e.
\begin{equation}
    y^{\text{e}}_t = \argmax_{y\ge 0} \left\{f_t(y) - c y\right\},    \quad\forall t\in [T], \label{eq:economic-standards}
\end{equation}
where $c>0$ is the cost per deployed staff. Then for the choice of the reward function \eqref{eq:chosen-reward-function},
\eqref{eq:economic-standards} implies
\begin{equation}
    y^{\text{e}}_t = 
    \begin{cases}
        \frac{d_t}{a} \log{\frac{a}{c}}, & \text{if }a > c, \\
        0, & \text{else}.
    \end{cases}
\end{equation}

In both cases, the shift plan can then be computed by minimizing the deviation between 
the supply $y_t$ and the desired supply $y_t'$, which is $y_t^s$ for service standards 
and $y_t^e$ for economic standards. 
Specifically, the shift plan is computed by solving the following mixed-integer quadratic program:

\begin{subequations}
\label{mip:seperate-demand-modelling}
\begin{align}
     \min \quad & \sum_{t=1}^{T} \left(y_t-y_t'\right)^2 \label{mip:seperate-quadratic-costfunc}\\
    \text{s.t.} \quad & y_{t} = \sum_{\tau=t-\delta + 1}^{t} x_{\tau}, && \text{for all } t\in [T], \\
    &y_t \le c, &&\text{for all } t\in [T],  \\
    &\sum_{t=1}^T x_t = sN, \\
    &z_t = \sum_{\tau=t-\delta-\beta +1}^{t} x_t, &&\text{for all } t\in [T], \\
    &z_t \le N, &&\text{for all } t\in [T],  \\
    &x_t \ge 0 &&\text{for all } t\in [T], \\
    &x_t \in \mathbb{Z} &&\text{for all } t\in [T].
\end{align}
\end{subequations}
We have chosen the quadratic objective function \eqref{mip:seperate-quadratic-costfunc}
instead of the sum of absolute deviations $\sum_{t=1}^{T} \left|y_t-y_t'\right|$ because the former 
penalizes large deviations more than the latter.

\section{Results}
Now we present the results of solving the mixed-integer convex program described in 
Section~\ref{sec:full-program}, using the open source SCIP optimization suite~\cite{BestuzhevaEtal2021OO}.
To this end, we reduce the problem to a mixed-integer linear program (MIP) by
first replacing the reward function $f_t$ defined in \eqref{eq:concr-prodmodel} with a piecewise linear 
concave approximation 
\begin{equation}\label{eq:piecewise-linear-approximation}
    f_t'(y_t) = \min \left\{ h_{t,1}(y_t), \dots, h_{t,k}(y_t) \right\}, 
\end{equation}
where $h_{t,1}, \dots, h_{t,k}$ are linear functions.
The resulting program can then be transformed into the following MIP:
\begin{subequations}
\label{mip:piecewise-linear}
\begin{align}
     \max \quad & \sum_{t=1}^{T} r_t\\
    \text{s.t.}\quad  & r_t \le h_{t,i}(y_t) \quad \text{for all } i\in [k], \; t\in [T], \\
    &\eqref{program:first-constraint} - \eqref{program:last-constraint}.
\end{align}
\end{subequations}
We note that, in order to obtain an accurate approximation of the reward function $f_t$
while not increasing the number of variables and constraints too much,
we may allow the number of linear segments $k$ in the piecewise linear approximation
\eqref{eq:piecewise-linear-approximation} to depend on the time $t$.

\subsection{Solving the MIP to produce a shift plan}
In the following, we will use the term \emph{shift agnostic optimum} for both the 
optimal \emph{value} as well as the optimal \emph{solution} of the optimization problem 
\eqref{problem:shift-agnostic-optimum}, whenever it is clear from the context which one is meant.
\Cref{fig:one-solution} shows a solution to the MIP formulation \eqref{mip:piecewise-linear} for 
$N=10$, $s=5$, $\delta=8$, $\beta=8$, $d_{\max}=10$, and $a=2$. 
The solution is reasonably close to the shift agnostic optimum, following the periodic demand pattern. However, 
the various shift constraints prevent the solution
from perfectly matching the shift agnostic optimum. In particular, the solution
does not have the smooth shape of the shift agnostic optimum, but is piecewise  
constant. Moreover, the solution does not reach the peak of the shift agnostic optimum  
in the middle of the week. In the following sections, we will investigate the impact 
of the various shift constraints on the quality of the solution and how relaxing them
helps to overcome these limitations.

\begin{figure}[!hbt]
\centering
\includegraphics[width=\textwidth]{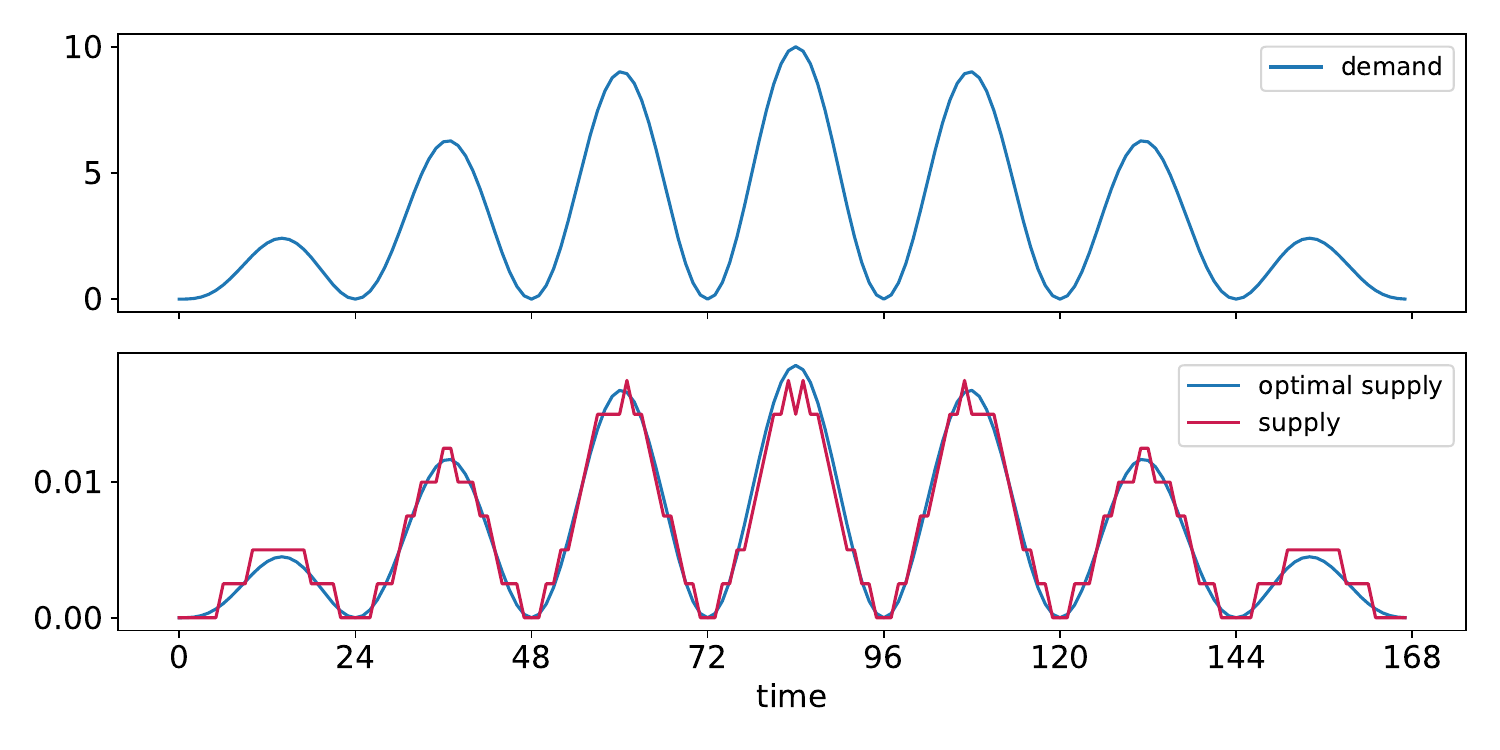}
\caption{A solution to the MIP formulation \eqref{mip} that is close to the shift agnostic 
optimum. Solved for $N=10$, $s=5$, $\delta=8$, $\beta=8$, $d_{\max}=10$, $a=2$.}
\label{fig:one-solution}
\end{figure}

\subsection{Impact of different shift constraints on the gap between solution and shift 
agnostic optimum}
In this subsection, we analyze the impact of various shift constraints on the
quality of the optimum feasible shift plan. Specifically, we investigate the impact of the
total number of drivers $N$, the number of shifts per driver $s$, and the shift
length $\delta$ on the \emph{\relGap{}} as defined in \eqref{def:relgap}.

\subsubsection{Impact of the total number of drivers}
Since every started shift is active for $\delta$ time steps, the active shifts
curve $\vec y = (y_t)_{t\in [T]}$ is the sum of $sN$ step functions. For a small
number of drivers $N$, it is therefore difficult to approximate the smooth shift
agnostic optimum $\vec y^* = (y^*_t)_{t\in [T]}$. As the number of drivers
increases, the approximation becomes more accurate, and the gap between the
solution $\vec y$ and the shift agnostic optimum $\vec y^*$ decreases.
\Cref{fig:different-vrh} illustrates this effect. For an increasing number of
drivers~$N$, we computed the corresponding optimal shift plan $\vec x$ and its 
supply curve $\vec y$ as well as the shift agnostic optimum $\vec y^*$.  For
comparability, we normalized each supply curve by dividing it by the total
working hours $sN\delta$.

For each number of drivers, we have computed the served trips using
\eqref{eq:chosen-reward-function}, and also the served trips for the shift
agnostic optimum. The gap between the two is then the \emph{lost trips} due to
shift restrictions.These lost trips, divided by the total
served trips due to the shift agnostic optimum, i.e. the \emph{\relGap} as
defined in \eqref{def:relgap}, are plotted in panel (b) of \Cref{fig:different-vrh}.  We
observe that the higher the number of drivers, the smoother the supply curve of
the optimal shift plan and the smaller the \relGap{}. Notice, however, that even
for as many as 200 drivers, the shift agnostic optimum is not reached. The
reason for this is that, since every driver needs to work $s$ shifts per week,
only $1/s$ of the total number of shifts can be active at the same time.  Thus,
the high peak in the middle of the week cannot be completely fulfilled.  This
limitation can be overcome by reducing the number of shifts per driver, as we
will see in the next section.

\begin{figure}[!ht]
\centering
\includegraphics[width=\textwidth]{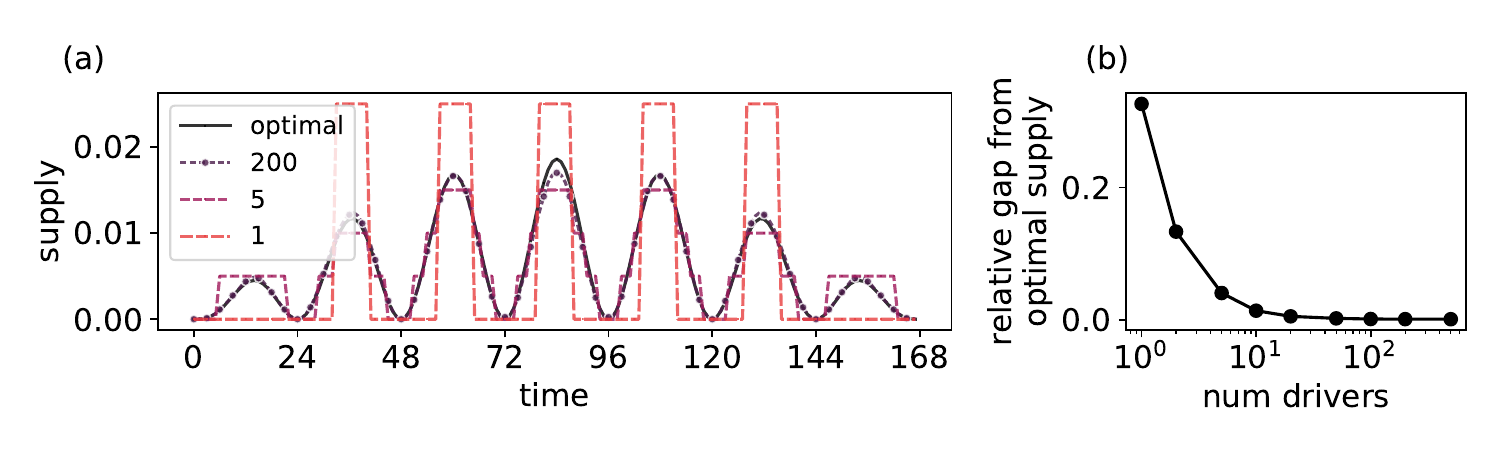}
\caption{As the number of drivers $N$ increases, the difference between the shift agnostic optimum 
and the optimal shift plan decreases. (a) The supply curve, for different values of total drivers (coloured dashed lines), 
as well as the shift agnostic optimum (solid black line). Each supply curve is
normalized, i.e., divided by the total working hours $sN\delta$.
(b) The \relGap{} compared to the shift agnostic optimum, which 
decreases as the number of drivers increases.}
\label{fig:different-vrh}
\end{figure}

\subsubsection{Impact of weekly shifts per driver}
If the total working time $sN\delta$ and the shift length
$\delta$ is fixed but both the number of drivers $N$ and the number of shifts
per driver $s$ are flexible, then it is beneficial to choose more drivers with fewer shifts;
see \Cref{fig:different_num_shifts}. As an intuitive example, consider two scenarios, 
one with $N$ drivers with $s$ shifts each, and another with $sN$ drivers with
only one shift each. Then any shift plan for the first scenario can also be realized
in the second scenario by assigning each of the $s$ shifts of a driver in the first scenario
to $s$ different drivers in the second scenario. 
Therefore, the optimal shift plan in the second scenario is at least as good as the optimal
shift plan in the first scenario. 
Moreover, while in the first scenario only $N$ shifts can be active at the same time,
in the second scenario $sN$ shifts can be active at the same time, allowing to reach 
higher peaks.

\begin{figure}[!ht]
\centering
\includegraphics[width=\textwidth]{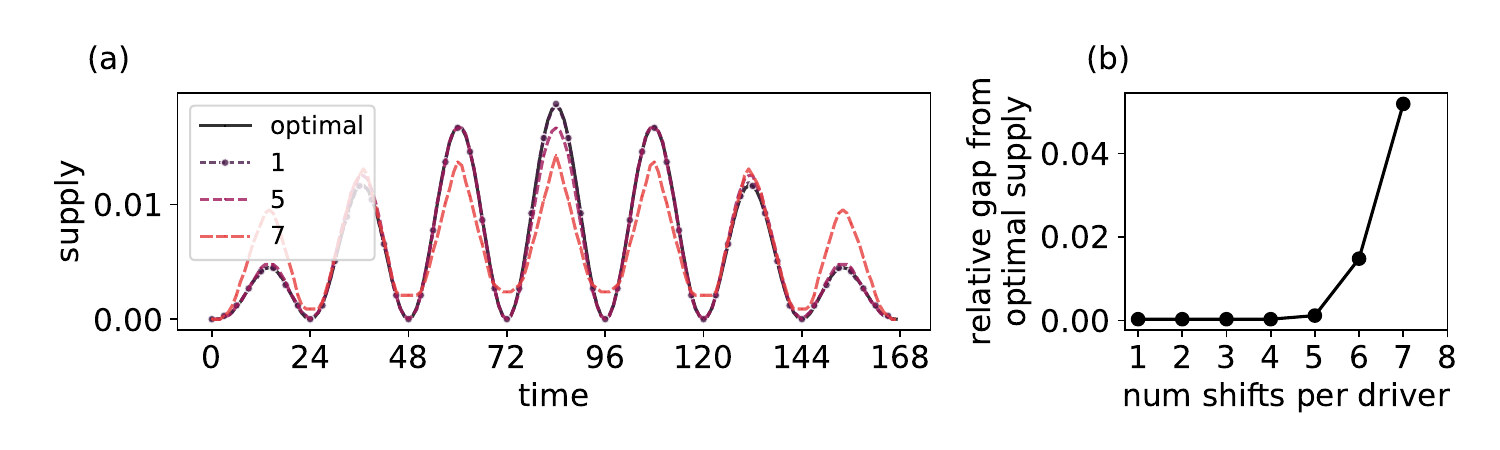}
\caption{As the number of shifts per driver $s$ decreases, the difference between the shift agnostic optimum
and the optimal shift plan decreases.}
\label{fig:different_num_shifts}
\end{figure}

\subsubsection{Impact of shift lengths}
Finally, if the total working time $sN\delta$ and the number of shifts
per driver $s$ is fixed but both 
the shift length $\delta$ and the number of drivers $N$ are flexible, 
then it is beneficial to choose more drivers with shorter shifts. In this  
way, we have more step functions with 
smaller support that can approximate
the smooth shift agnostic optimum better; see \Cref{fig:different_shift_length}. 
Consider, for example, five different
choices of shift length, $\delta\in \set{1,2,4,8,16}$. Then any solution with a shift length 
of $\delta = 2^i$, $i>0$, and $N$ drivers can also be realized by $2N$ drivers with
shift length $\delta' = \delta/2$ by substituting each shift of length $\delta$ with two
shifts of length $\delta'$. Therefore, the smaller $i$ is, the better the
resulting solution.

\begin{figure}[!ht]
\centering
\includegraphics[width=\textwidth]{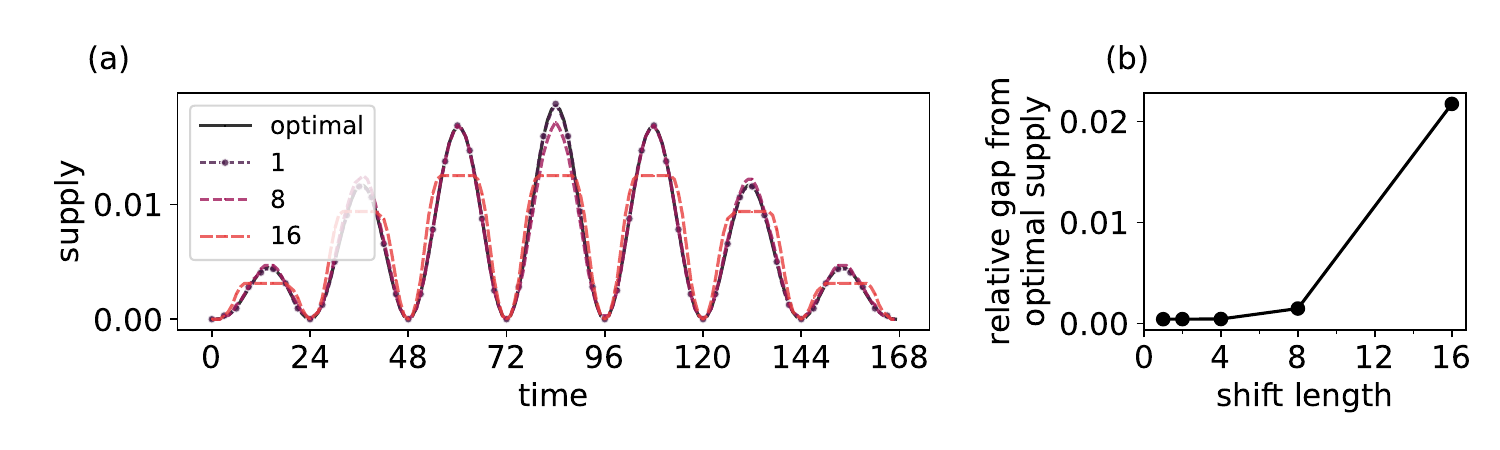}
\caption{As the shift length $\delta$ decreases, the difference between the shift agnostic optimum
and the optimal shift plan decreases.}
\label{fig:different_shift_length}
\end{figure}

\subsection{Comparison to approaches with a separate demand modelling step}
\label{sec:two-stage-solution}
\begin{figure}[!ht]
\centering
\includegraphics[width=0.8\textwidth]{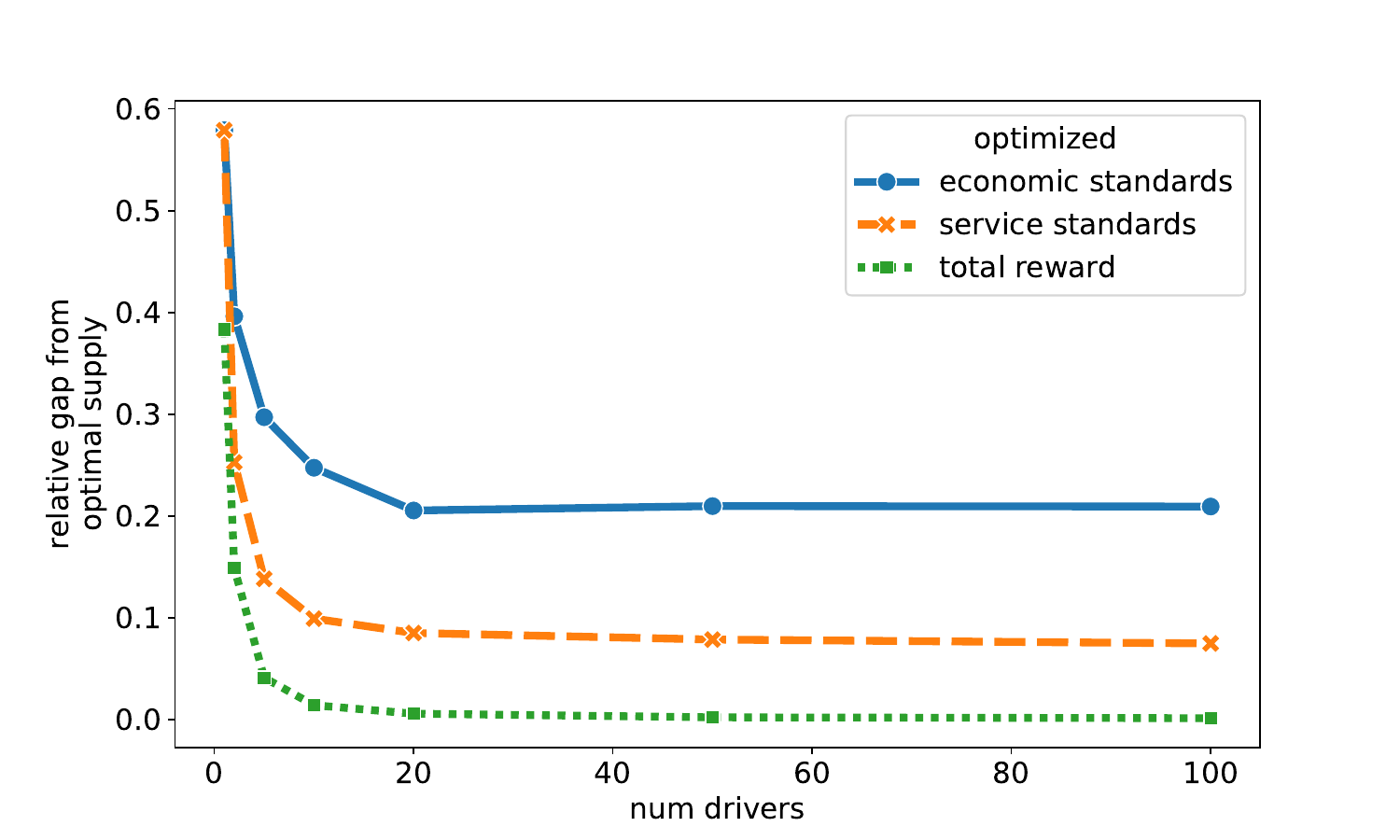}
\caption{Our shift planning method achieves higher total reward than approaches with 
a separate demand modelling step. We have plotted the \relGap{} 
\eqref{eq:shift-agnostic-optimum} for different values of the number of drivers $N$.
All three approaches achieve smaller \relGap{} as the number of drivers increases, but
our approach achieves the least \relGap{} for all values of $N$. For service standards,
the parameter $c$ was set to $0.8$, and for economic standards, the parameter $c$ was set to $1$.
Other parameters were set to $s=5$, $\delta=8$, $\beta=9$, $d_{\max}=0.75N$, $a=2$. 
} 
\label{fig:comparison_traditional_approaches}
\end{figure}

To demonstrate the benefit of combining the demand modelling
with the shift plan optimization in one step, we now compare the quality of the shift plans generated by our method with ones 
generated by the two approaches with separate demand modelling steps described in 
Section \ref{sec:separate-demand-modelling}. As before, we will use the \relGap{} as
the metric for these comparisons.

We see in Figure \ref{fig:comparison_traditional_approaches} that the quality of the 
shift plans generated by both of these approaches lead to significantly less total 
reward than the ones generated by our approach. This stems from the fact that our 
approach as described in the problem formulation \eqref{full-mip-ours} directly maximizes
the total served trips, thereby minimizing the \relGap{}. Both of the other approaches, 
by the very nature of having a separate demand modelling step, first compute the desired
supply, and then minimize the deviation from that. 
This two-step process leads to a loss of information, since
the deviation from the desired supply
does not capture the resulting lost revenue: The same deviation at different points in time
can result in different amounts of lost revenue depending on the demand pattern.
As a result the two-step optimization 
process does in general not maximize the total served trips.

\paragraph{Robustness of traditional approaches with internal parameters}
\begin{figure}[!ht]
\centering
\includegraphics[width=\textwidth]{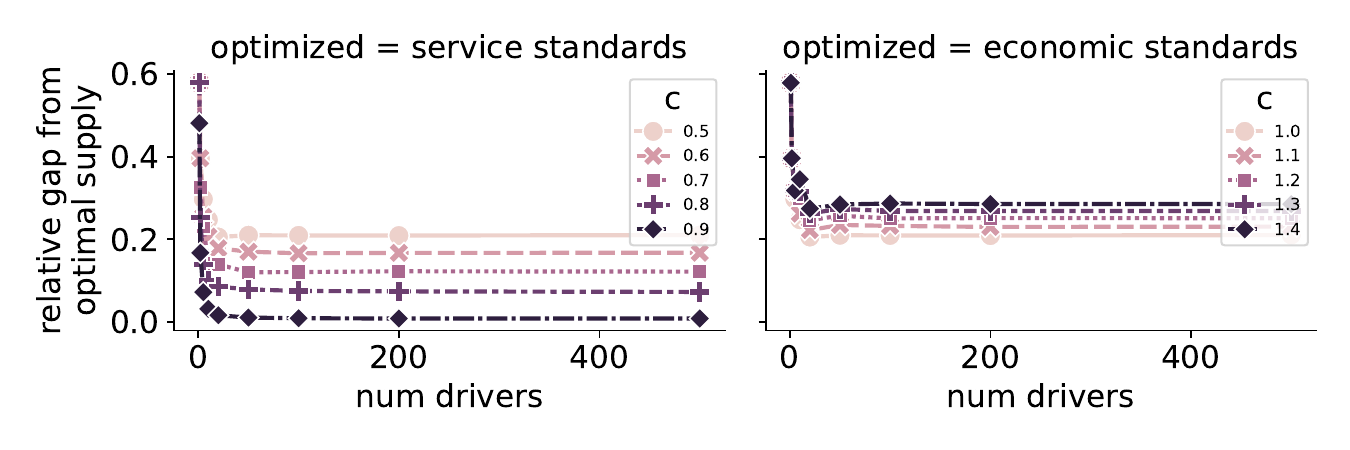}
\caption{Staff scheduling approaches with demand modelling for both service standards 
and economic standards are highly dependent on the choice of the internal parameters. We
have plotted the \relGap{}
for both approaches, for different choices of the internal parameter $c$.}
\label{fig:traditional_approaches_different_parameters}
\end{figure}

We also see in \Cref{fig:traditional_approaches_different_parameters} that the
quality of the shift plans is highly dependent on the choice of service
parameter in \eqref{eq:service-standards} for the service standards approach, 
and on the choice of the cost parameter in \eqref{eq:economic-standards} for the
economic standards approach. This demonstrates another advantage of our approach
of directly maximizing the total reward, namely that it is not dependent on any such
internal parameters.


\section{Outlook}
In this article we have introduced a novel approach to staff scheduling, where the total
reward over the planning period is directly maximized, instead of first computing a 
desired supply by a demand modelling step. For showcasing the benefits of our approach,
we chose a scenario with only one shift type and an 
exponential reward function, for the sake of simplicity. It will be interesting to study 
how our approach performs in more complex scenarios.

\paragraph{More complex reward function and shift types}
First of all, it would be interesting to apply our approach to scenarios with more than one shift type.
This would require generalizing the constraints
\eqref{eq:constraints-max-extended-shifts} and \Cref{lemma:shift-assignment}
to ensure that the optimized shift plan is rosterable.  
Also, we have so far not considered
breaks within a shift, but this can be easily accommodated by redefining active
shifts in \eqref{eq:active-shifts}. 

If the reward function is not concave in certain bounded 
subsets of its domain, often it is possible to approximate it with a concave 
approximation, e.g. by using the concave hull of a piecewise linear approximation. 
How the improvement of using our approach depends on the choice of the reward 
function is also an interesting question.

Further, it might be the case that the different shift types contribute in different ways to the
reward function. For example, full-time employees might be more experienced and thus
more productive than employees only working a few hours per week. 
In this case, instead of letting the reward function 
$f_t$ only depend on the aggregated supply $y_t = \sum_{i\in [k]}y_{t,i}$, 
one can use a reward function $f_t(y_{t,1}, \dots, y_{t,k})$ that directly depends on the 
supply of each shift type. If it is possible to express the reward function as
$f_t(y_{t,1}, \dots, y_{t,k}) = \sum_{i\in [k]} f_{t,i}(y_{t,i})$, and all of the functions 
$f_{t,i}$ are concave, then the results of this article can easily be extended to this 
more general case.

Instead of only maximizing the total reward, a company might also want to maximize other
metrics at the same time, for example, to achieve the right 
trade-off between total revenue and service quality. It would be interesting to extend 
our method to such use cases by applying techniques from multicriteria optimization.

\paragraph{Extending our approach to shift assignment}
In this article, we have limited ourselves
to producing a shift plan without assigning the shifts to individual employees. 
For only one shift type, as we have considered in \Cref{sec:application},
the proof of \Cref{lemma:shift-assignment} can readily be turned into
an algorithm for shift assignment. A generalization of \Cref{lemma:shift-assignment} 
and an algorithm for shift assignment to multiple shift types is left for future research.
We note that a different approach for assigning shifts to employees is to redefine shift types 
as described in \Cref{sec:problem-setting} so that each shift type
corresponds to shifts from a single employee. However, this comes at the cost of a larger 
optimization problem.

\section{Conclusion}
We have presented a novel approach to staff scheduling in this article, where the 
total reward over the planning period is directly maximized, instead of first computing
a desired supply by a demand modelling step and then minimizing the deviation from that.
We have shown that our approach leads to higher total reward than the traditional 
approaches. We have also presented a novel metric for evaluating the impact of 
constraints on the quality of a shift plan, the \relGap{}.

\section{Declaration of interests}
We acknowledge that Debsankha Manik and Rico Raber are employed at the ride-pooling operator
MOIA.

\section*{Acknowledgements}
We thank our colleagues at MOIA GmbH and MOIA Operations Germany GmbH for their valuable input, and especially Yılmaz Arslanoğlu, who co‑developed the initial prototype of the employee-scheduling solution presented in this paper.

\printbibliography
\end{document}